\documentclass[11pt]{article}
\usepackage[a4paper,margin=1in]{geometry}

\usepackage{graphicx}
\usepackage{color}
\usepackage{hyperref}

\usepackage{booktabs}
\usepackage{multirow}
\usepackage[utf8]{inputenc}
\usepackage{amsmath}
\usepackage{amsfonts}
\usepackage{amssymb}
\usepackage{amsthm}
\usepackage{authblk}

\allowdisplaybreaks

\usepackage{mathtools}
\DeclarePairedDelimiterX{\infdivx}[2]{(}{)}{#1 \,\delimsize\|\, #2}

\newcommand{\eps}{\varepsilon}
\newcommand{\R}{\mathbb{R}}

\newcommand{\norm}[2]{\left\| #1 \right\|_{#2}}

\newcommand{\Ep}[1]{\text{E}\!\left[#1 \right]}

\renewcommand{\epsilon}{\varepsilon}
\newcommand{\BO}[1]{{ O}(#1)}

\newcommand{\BTO}[1]{\tilde{ O}\left(#1\right)}

\newcommand{\BT}[1]{{\Theta}(#1)}
\newcommand{\BOM}[1]{\Omega(#1)}

\renewcommand{\Pr}[1]{\textnormal{Pr}\left[#1\right]}

\newtheorem{lemma}{Lemma}
\newtheorem{theorem}{Theorem}
\newtheorem{definition}{Definition}

\date{}

\begin{document}
\title{Similarity Search with Tensor Core Units\thanks{This work was partially supported by  UniPD SID18 grant, PRIN17 20174LF3T8, MIUR  "Departments of Excellence".
}}

\author[1]{Thomas D. Ahle}
\author[2]{Francesco Silvestri}

\affil[1]{
  IT University and BARC, Copenhagen, Denmark {thdy@itu.dk}}
\affil[2]{
  University of Padova, Padova, Italy {silvestri@dei.unipd.it}
}

\maketitle

\begin{abstract}
Tensor Core Units  (TCUs) are hardware accelerators developed for  deep neural networks, which efficiently support the  multiplication of two dense $\sqrt{m}\times \sqrt{m}$ matrices, where $m$ is a given hardware parameter. 
In this paper, we show that TCUs can speed up similarity search problems as well. We propose algorithms for the Johnson-Lindenstrauss  dimensionality reduction and for similarity join that, by leveraging TCUs, achieve a $\sqrt{m}$ speedup up with respect to traditional approaches.

\end{abstract}

\section{Introduction}\label{sec:intro}
Several hardware accelerators have been  introduced to speed up deep neural network computations, such as Google's Tensor Processing Units~\cite{jouppi2017datacenter} and NVIDIA's Tensor Cores~\cite{nvidia-volta}.
The most important feature of these accelerators is a hardware circuit to efficiently compute a small and dense matrix multiplication between two $\sqrt{m}\times\sqrt{m}$ matrices, where $m$ is a given hardware parameter.
On modern chips $m$ can be larger than $256$~\cite{jouppi2017datacenter}.
Matrix multiplication is indeed one of the most frequent operations in machine learning, and specialized hardware for supporting this operation can significantly reduce running times and energy requirements~\cite{JouppiYPP18}.
We refer to these accelerators as \emph{Tensor Core Units} (TCUs).

Recently, several studies have been investigating how to use TCUs in other domains.
For instance, TCUs have been used for scanning and prefix computations \cite{Dakkak19}, linear algebra primitives like matrix multiplication and FFT \cite{CSV20,LC20}, and graph problems \cite{CSV20}.
The key designing goal when developing TCU algorithms is to decompose the problem into several  small matrix multiplications of size $\sqrt{m}\times \sqrt{m}$, which are then computed on the accelerator.
Such algorithms also imply fast external memory algorithms, though not the other way around, since the matrix multiplication chip can be seen as a restricted cache~\cite{CSV20}.

The goal of this paper is to show that TCUs can  also speed up similarity search problems.  
As case studies, we propose TCU algorithms for the Johnson-Lindenstrauss  dimensionality reduction and for similarity join.
In both cases, our results improve the performance by a factor $\sqrt{m}$ with respect to state of the art approaches without hardware accelerators.

We analyze our algorithms on the \emph{$(m,\tau)$-TCU model}, which is a computational model introduced in \cite{CSV20} and capturing the main hardware features of TCU accelerators. In the $(m,\tau)$-TCU model, it is possible to compute the matrix multiplication between two matrices of size $\sqrt{m}\times\sqrt{m}$ in time $\tau$, where $m$ and $\tau$ are given parameters. 
In a traditional machine, without accelerators, we have $\tau=\BT{m^{3/2}}$
\footnote{Fast matrix multiplication algorithms require $O(m^{\omega/2})$ time with $\omega\in[2,3]$,~\cite{DBLP:conf/coco/Alman19}, but they exhibit poor experimental performance  than  traditional $\Theta(m^{3/2})$ algorithms.
}.
In contrast, with TCUs, we have  $\tau=O(m)$ (i.e., input size complexity) or even sublinear time under some assumptions.

The \emph{Johnson-Lindenstrauss (JL)} dimensionality transform reduces the dimension of a vector $x\in\R^d$ to roughly $k=\eps^{-2}\log(1/\delta)$ while preserving its norm up to a factor $1\pm\eps$ with probability at least $1-\delta$.
It is an important primitive in many learning algorithms, since it dramatically reduces the number of trained variables, while preserving important characteristics of the feature vectors, such as their pairwise inner products.
The JL transform can be represented as a multiplication of the input vector $x\in\R^d$ by a $k\times d$ matrix.
This naively takes time $\Omega(dk)$.
In this paper we use recent breakthroughs in dimensionality reduction techniques, combined with TCU's to reduce the time to $O(dk/\sqrt{m}+d+k^2\log^3\frac dk)$.
This is significant, since TCUs typically cut a factor $\sqrt{m}$ off matrix-matrix multiplication, but here we cut $\sqrt{m}$ off \emph{matix-vector} multiplication!
When $\sqrt{m}\ge k$ our dimensionality reduction takes time linear in the input dimension.
This improves upon even the famous ``Fast Johnson Lindenstrauss'' transform~\cite{ailon2006approximate}, which takes time $\Omega(d\log d+k^{2+\gamma})$ for any $\gamma>0$~\cite{ailon2009fast}, or $\Omega(d\frac{\log d}{\log m})$ with TCU optimized FFT~\cite{CSV20}.

The \emph{Similarity Join} on two sets $P$ and $Q$ of $n$ points each in $\R^d$, asks us to find all pairs $(x,y)\in P\times Q$ whose distance is below a given threshold $r$ (i.e., all near pairs).
Similarity join occurs in numerous applications, such as web deduplication and data cleaning. 
As such applications arise in large-scale datasets, the problem of scaling up similarity join for different metric distances is getting more important and more challenging. 
Exact similarity join is not possible faster than brute force~\cite{ahle2016complexity},
but by leveraging Locality Sensitive Hashing (LSH), we will develop a TCU approximate algorithm that, under some assumptions, finds all pairs in expected time $\BO{ (\frac{n}{\sqrt{m}})^\rho(\frac{|P \bowtie_r Q|d}{\sqrt{m}}  + n)}$,  where $|P \bowtie_r Q|$ is the number of near pairs.
When $\tau=\BO{m}$, the TCU algorithm exhibits a $\sqrt{m}$ speedup with respect to traditional approaches (even those based on LSH).

\section{Preliminaries}\label{sec:prelim}
\subsection{The TCU model}
The $(m,k)$-TCU model is a RAM model with an instruction to multiply two dense matrices of size $\sqrt{m}\times \sqrt{m}$ in time $\tau$, where $m$ and $\tau$ are given parameters depending on the underline platform.\footnote{The model in \cite{CSV20} is slightly different, and we use here a simplified version for the clarity of exposition.}
It is reasonable to assume that $\tau = \BO{m}$, that is  matrix multiplication takes linear time: indeed, on TCUs, the cost of the operation is upper bounded by the time for reading/writing the $\sqrt{m}\times \sqrt{m}$ matrices, while the cost of the $m^{3/2}$ elementary products is negligible due to the high level of parallelism inside TCU accelerators (e.g., systolic array).
Moreover, under some conditions on high bandwidth connections, we might have $\tau$ to be even sublinear (e.g., $\BO{\sqrt{m}}$).
We recall a result from \cite{CSV20} that will be used later:
\begin{theorem}\label{thm:mmult}
Let $A$ and $B$ be two matrices of size $p\times r$ and $r\times q$ with $p,r,q\geq \sqrt{m}$, then there exists an algorithm for computing $A\cdot B$ on a $(m,\tau)$-TCU model in time
$
\BO{prq m^{-3/2}\tau}. 
$ 
\end{theorem}

\subsection{Johnson-Lindenstrauss dimensionality reduction}

We say a distribution over random matrices $M\in \R^{k\times d}$ is a $(\eps,\delta)$-Johnson-Lindenstrauss (JL) distribution, if we have
$\Pr{|\|Mx\|_2 - 1|\le\eps}\ge1-\delta$  for all unit vectors $x\in\R^d$
In this section we will note some definitions and lemmas related to building and combining random matrices in ways related to JL distributions.
The first property was introduced by Kane and Nelson~\cite{kane2014sparser}:
\begin{definition}[JL-moment property]
   We say a distribution over random matrices $M\in \R^{k\times d}$ has the $(\eps,\delta,p)$-JL-moment property, when $E[\|Mx\|_2^2]=1$ and
   $
      \left(\Ep{\left|\norm{Mx}2^2-1\right|^p}\right)^{1/p}
      \le \eps \delta^{1/p}
  $
   for all $x\in \R^d$, $\|x\|_2=1$.
\end{definition}
A distribution with the $(\eps,\delta,p)$-JL-moment property is $(\eps,\delta)$-JL because of Markov's inequality:
$
   \Pr{|\|Mx\|_2-1|>\eps}
   \le
   \Ep{\left|\norm{Mx}2^2-1\right|^p}/\eps
      \le \delta.
$

An interesting property of the JL Moment Property is related to the tensor product of matrices.
The tensor (or Kronecker) product between two matrices $A\in\R^{m\times n}$ and $B\in\R^{k\times\ell}$ is defined as below.
In particular, if we take the tensor product $I_k\otimes A$, where $I_k$ is the $k\times k$ identity matrix, we get a $km\times kn$ block matrix with $A$ on the diagonal:

   \begin{align*}
      A\otimes B = \begin{bmatrix} A_{1,1} B & \cdots & A_{1,n}B \\ \vdots & \ddots & \vdots \\ A_{m,1} B & \cdots & A_{m,n} B \end{bmatrix}
      ,
      \quad
   I_k \otimes A =
   \begin{bmatrix}
      A & 0 & \cdots & 0 \\
      0 & A & \ddots & \vdots \\
      \vdots & \ddots & \ddots & 0 \\
      0 & \cdots & 0 & A
   \end{bmatrix}.
\end{align*}

The tensor product relates to the JL-moment property by the following simple lemma from~\cite{ahle2020oblivious}:
\begin{lemma}[JL Tensor lemma]\label{lem:jl-tensor}
   For any matrix,
   $Q$, with $(\eps, \delta, p)$-JL moment property,
   $I_k \otimes Q$ has $(\eps, \delta, p)$-JL moment property.
\end{lemma}

By the simple property $A\otimes B = (I\otimes B)(A\otimes I)$ this lemma allows studying the JL properties of general tensor products, as long as we can also handle matrix products.
The following generalization of the JL Moment Property will be key to doing exactly that:

\begin{definition}[$(\eps,\delta)$-Strong JL Moment Property]\label{defn:strong-jl-moment}
   Let $\eps,\delta\in[0,1]$.
   We say a distribution over random matrices $M\in \R^{m\times d}$ has the $(\eps,\delta)$-Strong JL Moment Property, when
   $\Ep{\norm{Mx}2^2}=1$ and
   $\left(\Ep{\left|\norm{Mx}2^2-1\right|^p}\right)^{1/p}
   \le \frac{\eps}e\sqrt{\frac{p}{\log1/\delta}}$,
   for all $x\in \R^d$, $\|x\|_2=1$ and all $p$ such that $2 \le p\le \log1/\delta$.
\end{definition}

Note that the  $(\eps,\delta)$-Strong JL Moment Property implies the $(\eps,\delta,\log1/\delta)$-JL Moment Property, since then $\eps \delta^{1/p} = \eps/e$.
Similarly, having the $(\eps\sqrt{2/e},\delta,p)$-JL-moment property for all $p\in[2,\log1/\delta]$ implies the Strong JL Moment Property, since $\delta^{1/p} \le \frac1{\sqrt{2e}}\sqrt{\frac{p}{\log1/\delta}}$.

The key workhorse is the following lemma by Ahle and Knudsen~\cite{ahle2019almost}.
Note that the original lemma required the  $(\eps/(C_0\sqrt{k}), \delta)$-Strong JL Moment Property, but a quick scan of the proof shows that $(\eps/(C_0\sqrt{i}), \delta)$-Strong suffices.
\begin{lemma}[JL Product lemma]\label{lem:jl-product}
    There exists a universal constant $C_0$, such that, 
    for any constants $\eps, \delta \in [0,1]$ and positive
    integer $k \in \mathbb Z_{> 0}$. If
    $
        M^{(1)} \in \mathbb R^{d_2 \times d_1}, \ldots, 
        M^{(k)} \in \mathbb R^{d_{k+1} \times d_k}
    $
    are independent random matrices satisfying the 
    $(\eps/(C_0\sqrt{i}), \delta)$-Strong JL Moment Property,
    then the matrix $M = M^{(k)} \cdot \ldots \cdot M^{(1)}$
    has the  $(\eps, \delta)$-Strong JL Moment Property.
\end{lemma}
Intuitively this says that combining $k$ JL reductions, we don't get an error of $\eps k$, as we would expect from the triangle inequality, but only $\eps\sqrt{k}$, as we would expect from a random walk.

\subsection{Locality Sensitive Hashing} 
Much of recent work on similarity search and join has focused on Locality Sensitive Hashing:  at a high level,  similar points (i.e., with distance $<r$) are more likely to collide  than far points (i.e., with distance $>cr$ for a given approximation factor $c$). Formally, an LSH is an {$(r,cr,p_1,p_2)$-sensitive hashing scheme:

\begin{definition}\label{def:LSH}
Fix a distance function $D: \mathbb{U}\times \mathbb{U} \rightarrow {\bf R}$.
For positive reals $r,c,p_1,p_2$, and $p_1 > p_2, c > 1$, a family of functions $\mathcal{H}$ is \emph{$(r,cr,p_1,p_2)$-sensitive} if for uniformly chosen $h\in \mathcal{H}$ and all $x, y\in \mathbb{U}$:
\begin{itemize}
	\item If $D(x,y) \leq r$ then $\Pr{h(x)=h(y)}\geq p_1$;
	\item If $D(x,y) \geq cr$ then $\Pr{h(x)=h(y)}\leq p_2$.
\end{itemize}
We say that $\mathcal{H}$ is \textit{monotonic} if $\Pr{h(x)=h(y)}$ is a non-increasing function of the distance function $D(x,y)$.
\end{definition}
LSH schemes are characterized by the $\rho=\log_{p_2} p_1$ value, with $\rho\in[0,1]$: small values of $\rho$ denote LSHs that well separate near points from far points. Term $c$ is the approximation factor.

\section{Dimensionality Reduction}\label{sec:dr}
We will describe a construction of a matrix $M\in\R^{k\times d}$ which is $(\eps,\delta)$-JL as described in the preliminaries, and for which there is an efficient algorithm for computing the matrix vector product $Mx$ on a TCU.
We first give a general lemma describing the construction, then show how it applies to TCUs:
\begin{lemma}\label{jlconstr}
Let $T(a,b,c)$ be the time for multiplying two matrices of size  $(a\times b)$ and $(b\times c)$.
   For a  constant $C>0$ and for any $d, \eps, \delta>0$,
   there exists a matrix $M\in\R^{k \times d}$, with $k=\lceil C\eps^{-2}\log1/\delta\rceil$, such that
$
      |\|Mx\|_2 - \|x\|_2| \le \eps\|x\|_2
 $
 for any $x\in\R^{d}$
   with probability $1-\delta$ 
   (i.e., $M$ is $(\eps,\delta)$-JL).
   The multiplication $Mx$ can be computed in time
 $
      \sum_{i=1}^{\ell}
      T(ik,\, \zeta ik,\, \zeta^{\ell-i})
  $
   for any  $\zeta>1$ and $\ell$ such that $\zeta^\ell = d/k$.
\end{lemma}

Note that, depending on the speed of the rectangular matrix multiplication, it might be beneficial to pick different values for $\zeta$.

\begin{proof}
   We define the JL transformation by the following matrix:
   \begin{align*}
      M = (I_{r_\ell} \otimes A_\ell) \cdots (I_{r_1} \otimes A_1) \in \R^{r_m k_\ell \times r_1 c_1}
      ,
   \end{align*}
   where $r_1, \dots, r_\ell$ is a sequence of positive integers, $I_r$ is the $r\times r$ identity matrix, and 
   $A_1, \dots, A_{\ell-1}$ are independent $k_i\times c_i$ matrices, where $A_i$ has the $(\eps/(C_0\sqrt{i}),\delta)$-Strong JL Moment Property (SJLMP).
   By Lemmas~\ref{lem:jl-tensor} and \ref{lem:jl-product} we get that the tail
   $(I_{r_{\ell-1}} \otimes A_{\ell-1}) \cdots (I_{r_1} \otimes A_1) \in \R^{r_m k_\ell \times r_1 c_1}$
   has the $(\eps/\sqrt{C_0}, \delta)$-SJLMP.
   We further assume $A_\ell$ has the $(\eps/(\sqrt{2C_0}), \delta)$-SJLMP.
   Again by Lemmas~\ref{lem:jl-tensor} and \ref{lem:jl-product} we get that
   $M$ has the $(\eps, \delta)$-SJLMP, and thus $M$ is a JL reduction as wanted.

   Next we prove the running time of the matrix-vector multiplication.
   The key is to note that
   $I\otimes A$ is the ``block identity matrix'' with A copied along the diagonal.
   The following figure should give some some intuition:
   \begin{align*}
      (I_{r_i}\otimes A_i)x
      =
      \substack{r_i\\\text{blocks}}\left\{
      \begin{bmatrix}
         \substack{k_i}\big\{\smash[b]{{\underbrace{ A_i }_{c_i}}}&&\\
         &A_i&\\
         &&A_i
      \end{bmatrix}
      \right.
      x
      \simeq
   A_i \left.\begin{bmatrix}x_1 & \dots & x_{r_i}\end{bmatrix}\right\}\substack{c_i}
   =
   \left.\begin{bmatrix}y_1 & \dots & y_{r_i}\end{bmatrix}\right\}\substack{k_i}.
   \end{align*}

   By splitting $x$ into $r_i$ blocks, the multiplication $(I_{r_i}\otimes A)x$ corresponds to reducing each block of $x$ by identical JL matrices.
   Repeating this process for a logarithmic number of steps, we get the complete dimensionality reduction.

   To make sure the matrix sizes match up, we have
   \begin{align*}
      d = r_{1} c_{1} ,\quad
      r_{1} k_{1} = r_{2} c_{2} ,\quad
      r_{2} k_{2} = r_{3} c_{3} ,\quad
                  \dots ,\quad
      r_{\ell-1} k_{\ell-1} = r_{\ell} c_{\ell}, \quad
      r_\ell k_\ell = k
      .
   \end{align*}
   We will define $k=\lceil C\eps^{-2}\log1/\delta\rceil$, $k_{i<\ell}=i k$, $k_\ell=k$,
   $c_1=k\zeta$,
   $c_{i>1} = \zeta k_{i-1}$,
   $r_i = \zeta^{\ell-i}$ and
   $\ell=\frac{\log(d/k)}{\log\zeta}$ such that $c_1r_1 = k\zeta^\ell = d$.
   The constant $C$ depends on the constant of the JL lemma we use for the individual $A_i$, but in general $10 C_0^2$ will suffice, where $C_0$ is the constant of Lemma \ref{lem:jl-product}.

   Recall the assumption that rectangular multiplication takes time $T(a,b,c)$, and hence
   the $i$th step thus takes time $T(k_i, c_i, r_i)$.
   Adding it all up we get
   \begin{align*}
      \sum_{i=1}^\ell T(k_i, c_i, r_i)
      &= T(k, \zeta k (\ell-1), 1) + \sum_{i=1}^{\ell-1} T(ik, \zeta k\max(1,i-1), \zeta^{\ell-i})
   \end{align*}
which is then upper bounded by  $\sum_{i=1}^{\ell} T(ik, \zeta ik, \zeta^{\ell-i})$. The claim follows.
\end{proof}

By the above theorem and by using the  matrix multiplication algorithm of Theorem \ref{thm:mmult}, we get the following theorem.

\begin{theorem}\label{crl:jl}
  For any $d, \eps, \delta>0$, there exists a $(\eps,\delta)$-JL matrix $M\in\R^{k \times d}$ such that the product $Mx$ can be computed in time
  $
      \BO{(dk + k^2 \sqrt{m} \log^3\tfrac dk)\, \tau m^{-3/2}},
  $
    on the $(m,\tau)$-TCU model, assuming $k\geq \sqrt m$.
\end{theorem}

\begin{proof}
   By theorem \ref{thm:mmult} it is possible to multiply to matrices $p\times r$ and $r \times q$ in time $O(pqr \tau / m^{3/2})$ if $p,q,r\geq \sqrt{m}$.
   By assumption $k \ge \sqrt{m}$, so we get that the upper bound in Lemma \ref{jlconstr} becomes
   \begin{align*}
      \sum_{i=1}^{\ell} T(ik, \zeta ik, \zeta^{\ell-i})
         &\le \tau m^{-3/2} k^2 \zeta \sum_{i=1}^{\ell} i^2 \max(\zeta^{\ell-i},\sqrt{m})
         .
   \end{align*}
   We bound the maximum by the sum:
   \begin{align*}
         \sum_{i=1}^{\ell} i^2 \max(\zeta^{\ell-i},\sqrt{m})
         &\le
         \sum_{i=1}^{\ell} i^2 (\zeta^{\ell-i} + \sqrt{m})
         &\le
         \zeta^\ell \sum_{i=1}^{\infty} i^2 \zeta^{-i}
         +\sqrt{m} \sum_{i=1}^{\ell} \ell^2
         &=
         \zeta^\ell \frac{\zeta(\zeta+1)}{(\zeta-1)^3}
         +\sqrt{m} \ell^3
   \end{align*}
   A sharper analysis can change the last term to $\sqrt{m}\log^3 m$, but that doesn't dominate unless $m\ge d/k$, so it is ultimately fruitless.
   Putting it all together we get
   \begin{align*}
      \tau m^{-3/2} k^2 \zeta (
         \zeta^\ell\, \tfrac{\zeta(\zeta+1)}{(\zeta-1)^3}
         +\sqrt{m}\, \ell^3)
      = O(
      \tau m^{-3/2} k^2 (
      \tfrac{d}{k}
      +\sqrt{m} \log^3 \tfrac dk)
         )
   \end{align*}
   where we took $\zeta=2$ and recalled $\zeta^\ell = d/k$ by definition of $\ell$.
   That is what we wanted to prove.
\end{proof}

In particular, for $\tau=O(m)$ it takes time $O(dk/\sqrt{m} + k^2\log^3\frac dk)$.
If $\sqrt m > k$ we can ``pad'' the construction by increasing $k$ to $\sqrt{m}$ and simply throw away the unneeded rows.
The running time is then $O(d + k^2\log^3\frac dk)$.

We observe that if $\tau = O(m)$, and $d$ dominates $k^2$, then we get time $O(dk/\sqrt{m}))$, which improves a factor $\sqrt{m}$ over a standard application of the standard JL transform in the case of dense vectors, and for $m\approx k$ this even improves upon the so-called ``Fast JL transform''~\cite{ailon2006approximate}.

Finally, we note the following extra properties of the construction:
\begin{enumerate}
   \item In the case of sparse vectors, where many blocks of $x$ are empty, we can skip them in the computation.
   \item The computation can be easily parallelized, with different blocks of $x$ being reduced on different machines.
   Our construction also implies a $O(dk/\sqrt{m})$ upper bound in the external memory model.
   \item Our construction improves upon the standard matrix-vector multiplication for JL, even in the standard RAM model, by using the Coppersmith-Winograd method for fast matrix multiplication.
      In particular we can do JL in time $dk^\eps+k^{2+\eps}$ if matrix multiplication takes time $n^{2+\eps}$.
   \item The construction works with any distribution of matrices that have the Strong JL Moment Property. This means we can use random $\pm1$ matrices or even $\eps$-Sparse JL matrices.
\end{enumerate}

\section{Similarity Join}\label{sec:sj}
In this section we will study the similarity join problem, which is defined as follows:
given two sets $P$ and $Q$ of $n$ points each in $\mathbb{R}^d$ and a distance measure $D:\mathbb{R}^d \rightarrow R^+_0$, compute the set $P \bowtie_r Q = \{(x,y): x\in P, y\in Q, D(x,y)\leq r\}$. 

We consider distance measures that can be computed with an inner product on a suitable transformation of the two points (e.g, kernel function): specifically, for each pair $x,y\in \mathbb{R}^d$, there exist two functions $f,g:\mathbb{R}^d \rightarrow \mathbb{R}^{d'}$ such that $D(x,y) = f(x)\cdot g(y)$, where $\cdot$ denotes the inner product.
We refer to a distance function that satisfies this property as \emph{ip-distance}.
For the sake of simplicity, we assume $d'=\BT{d}$.
Notable examples of ip-distances are Hamming,  squared $L_2$ distance, and cosine similarity.
For Hamming, we have $f(x)=(x_0, 1-x_0, x_1, 1-x_1,\ldots, x_{d-1}, 1-x_{d-1})$ and $g(x)=(1-y_0, y_0, 1 -y_1, y_1 \ldots, 1-y_{d-1},y_{d-1})$.
For the squared $L_2$ distance, we have 
$f(x)=(x^2_0, 1, -2x_0, x^2_1, 1, -2x_1\ldots, x^2_{d-1}, 1, -2x_{d-1})$  and $
g(x)=(1, y^2_0, y_0, 1, y^2_1, y_1 \ldots, 1, -y_{d-1}^2, y_2)$.
Finally, for cosine similarity, we have $f(x)=g(x)=x/||x||_2$.

The simplest way to exploit TCUs is a brute force approach, where all pair distances are computed. As ip-distance computations can be translated into inner products, we can reduce the similarity join problem to a simple matrix multiplication between two $n\times d'$ matrices $F_P$ and $G_Q$: $F_P$ and $G_Q$ are the matrices representing, respectively, the sets $\{f(p), \forall p\in P\}$ and $\{g(q), \forall q\in Q\}$.
By exploiting TCUs, we can compute $P\cdot Q^T$ in time $\BO{dn^2 m^{-3/2} \tau}$.

We now show a more efficient approach that uses LSH for reducing the number of candidate pairs for which we have to compute the distance. 
The proposed algorithm finds all $P \bowtie_r Q$ pairs in expectation, but it can be easily modified to return all near pairs with high probability by running $\BO{\log n}$ instances of the algorithm and merging the results.

The standard LSH approach for similarity join (see e.g. \cite{Gionis_VLDB99,PPSS17}) requires to partition the points in $P\cup Q$ into buckets using an $(r, cr, p_1, p_2)$-sensitive monotone LSH. Then, a brute force algorithm is used for searching similar pairs  within each bucket. The procedure is then repeated $L$ times with independent LSHs to guarantee that all near pairs are found.
The LSH is usually set so that $p_2=1/n$, which implies that each point collides  once (in expectation) with a point at distance larger than $cr$ (i.e., a far point), while $L$ is set to $\BTO{p_1^{-1}}=	\BTO{p_2^{-\rho}} = \BTO{n^\rho}$ to guarantee that each near pair is found once (in expectation).

As for similarity join in   the external memory model \cite{PPSS17}, 
we can improve the performance in the TCU model by increasing the value of $p_2$ (i.e.,  by allowing for more collisions between far points), which implies that the number $L$ of repetitions decreases since $L=p_1^{-1} = \BTO{p_2^{-\rho}}$.
 We observe that a TCU unit can multiply two matrices of size $\sqrt{m'}\times \sqrt{m'}$ in a TCU$(m,\tau)$  in $\tau$ time for each $m'\leq m$, and we exploit this fact by increasing the number of collisions with far points.
We  set $p_2=m^{3/2}/(\tau n)$:  each point collides in expectation with at most $m^{3/2}/\tau$ far points, but the overhead due to the respective inner products do not dominate the running time.

As an LSH is usually given as a black box $\mathcal{H'}$ with fixed probability values $p_1'$ and $p'_2$, we can  get the desired probability  $p_2=m^{3/2}/(\tau n)$ by concatenating $k= \log_{p'_2} p_2$ hash functions. However, if $k$ is not an integer, the rounding gives $L= \BO{n^\rho p_1^{-1}}$.
A more efficient approach has been recently proposed in \cite{A20} that uses $L_{high}$ hash tables by concatenating $\lceil k\rceil$ LSHs $\mathcal{H'}$, and $L_{low}$ hash tables by concatenating $\lfloor k\rfloor$ LSHs $\mathcal{H'}$, and where  $L=L_{low}+L_{high}= \BO{n^\rho p_1^{-(1-\rho)}}$. 
The right values of $L_{low}$ and $L_{high}$ depend on the decimal part of $k$.

The above algorithm gives the following result (proof in appendix). 

\begin{theorem}\label{thm:sh}
Given two sets $P,Q\subset R^{d}$ of $n$ points, with $n,d\geq \sqrt{m}$, 
a threshold value $r>0$, and a $(r,c,p_1,p_2)$-sensitive LSH, then the set $P \bowtie_r Q$  for an ip-distance can be computed  on a TCU$(m,\tau)$ in expected time:
$$
\BO{  p_1^{\rho-1}  (n \tau m^{-3/2})^\rho \left(\frac{|P \bowtie_r Q|\tau}{ m^{3/2}}  + n\right)+  \tau m^{-3/2} |P \bowtie_{cr} Q| }.$$
\end{theorem} 
\begin{proof}
We initially prove the theorem by assuming that our LSH matches $p_2 =  m^{3/2} /(n \tau)$ and $p_1=p_2^\rho$.
Let  $P_{i,j}$ be the set of points in $P$ with hash value $j$ under the $i$-th LSH, for each $1\leq i \leq L$ and let $p_{i,j}=|P_{i,j}|$
(similarly for $Q_{i,j}$ and $q_{i,j}$ ).
For each $p_{i,j}>0$, we can assume $q_{i,j}\geq \sqrt{m}$ since there are in expectation $np_2=m^{3/2}/\tau = \BOM{\sqrt{m}}$ far points per each point in $P_{i,j}$; equivalently, for  each $q_{i,j}>0$
we can assume $p_{i,j}\geq \sqrt{m}$.
If the brute force join within a bucket is carried out with the  matrix multiplication in Theorem \ref{thm:mmult}, we have that the cost of the algorithm is
$
T_{\textnormal{sim join}} = \sum_{i=1}^{L}\sum_{j:p_{i,j}>0}  d p_{i,j} q_{i,j} m^{-3/2} \tau.
$
For given values of $i$ and $j$, the $p_{i,j}q_{i,j}$ pairs can be split into three categories: the $R_{i,j}$ pairs with distance in $[0,r]$, the $C_{i,j}$ pairs with distance in $(r, cr]$, and the $F_{i,j}$ pairs  with distance $>cr$ (far pairs). We have  $p_{i,j}q_{i,j} = R_{i,j}+C_{i,j}+F_{i,j}$.
Since  pairs with distance in $[0,r]$ collide with probability at most 1, pairs with distance in $[r, cr]$ collide with probability at most $p_1$ (for the monotonicity of LSH), and pairs with distance $>cr$ collide with probability at most $p_2$, we have 
$E[R_{i,j}+C_{i,j}+F_{i,j}]\leq |P \bowtie_r Q| + |P \bowtie_{cr} Q|p_1 + n p_2$.
Therefore, the expected value of the running time is:
\begin{align}
E[T_{\textnormal{sim join}}] &= E\left[\sum_{i=1}^{L}\sum_{j:p_{i,j}>0}  d (R_{i,j}+C_{i,j}+F_{i,j})]m^{-3/2}\tau \right] \nonumber\\
&\leq  d L m^{-3/2}\tau  E\left[\sum_{j:p_{i,j}>0}  \left( R_{i,j} + C_{i,j} + F_{i,j}\right)\right]\nonumber\\
&\leq L  m^{-3/2}\tau  \left( |P \bowtie_r Q| + |P \bowtie_{cr} Q |p_1 +  p_2 n\right)\label{eq:1}
\end{align}
from which we get the theorem since  $L=p_1^{-1}=p_2^{-\rho}$ and $p_2=m^{3/2}/(\tau n)$.

Finally, we consider we are given an LSH $\mathcal H$  where 
$p_2 >  m^{3/2} /(n \tau)$. If $k=\log_{p_2} (m^{3/2} /(n \tau))$ is an integer, then it suffices to construct another LSH by concatenating $k$ copies with collision probability $p_2$ (i.e., $\mathcal H^k$).
If $k$ is not an integer, we use the approach in \cite{A20}: we construct $L_{high}$ hash tables by concatenating $\lceil k\rceil$ LSHs $\mathcal{H'}$ and $L_{low}$ hash tables by concatenating $\lfloor k\rfloor$ LSHs $\mathcal{H'}$.
We have that $L=L_{low}+L_{high}= \BO{n^\rho p_1^{-(1-\rho)}}$, which implies a  multiplicative factor $\BO{p_1^{\rho-1}}$ in the previous upper bound in Equation \ref{eq:1}.
The right values of $L_{low}$ and $L_{high}$ depend on the decimal part of $k$ and we refer to \cite{A20} for the exact values.
\end{proof}

When $\tau=\BO{m}$, there are at least $n\sqrt{m}$ near pairs, and the number of pairs with distance in $[r,cr]$ is at most linear with the number of near pairs (which happens in several datasets \cite{PPSS17}), the cost is $\BO{  p_1^{\rho-1} (n /\sqrt{m})^\rho {|P \bowtie_r Q|}/{ \sqrt{m}}  }$, a factor at least $\sqrt{m}$ faster than an LSH solution without TCU (e.g., $\BO{  p_1^{\rho-1} n^\rho {|P \bowtie_r Q|} }$).

\section{Conclusion}\label{concl}
In this paper, we have investigated from a theoretical point of view how to exploit TCU accelerators for similarity search problems, showing a $\sqrt{m}$ improvement over algorithms for traditional architectures.
As future work, we plan to experimentally evaluate our algorithms on common TCU accelerators, such as the GPU Nvidia Tesla. 
 
\bibliographystyle{plain}
\bibliography{biblio}

\end{document}